\documentclass[runningheads, letter]{llncs}
\usepackage{amssymb}
\usepackage{graphicx}
\usepackage{extarrows}
\usepackage{amssymb,amscd}
\usepackage[all]{xy}
\usepackage{longtable}
\usepackage[latin1]{inputenc}
\usepackage{tikz}
\usetikzlibrary{shapes,arrows}
\newcommand{\keywords}[1]{\par\addvspace\baselineskip
\noindent\keywordname\enspace\ignorespaces#1}

\newcommand{\fq}{{\mathbb F}_{q}}
\newcommand{\mc}{\mathcal{C}}

\begin{document}
\mainmatter
\title
{A class of twisted generalized Reed-Solomon codes\thanks{The research of Jun Zhang was supported  by the National Natural Science Foundation of China under the Grant 11971321 and the National Key Research and Development Program of China under Grants 2018YFA0704703. The research of Zhengchun Zhou was supported by the National Natural Science
		Foundation of China under Grant 62071397 The research of Chunming Tang was supported  by the National Natural Science Foundation of China under the Grant 11871058.}}

 \author{Jun Zhang\inst{1} %
 	 \and Zhengchun Zhou\inst{2}%
 	 \and Chunming Tang\inst{3}
 }
 \institute{
 School of Mathematical Sciences, Capital Normal University\\
 Beijing 100048, China\\
 Email: {\tt junz@cnu.edu.cn}
  \and
School of Mathematics, Southwest Jiaotong University\\
 Chengdu, 610031, China \\
 Email: {\tt zzc@swjtu.edu.cn}
 \and 
 The School of Mathematics and Information, China West Normal University,\\
  Nanchong, 637002, China\\
  Email: {\tt tangchunmingmath@163.com} 
 }

\maketitle

\begin{abstract}
  Let $\fq$ be a finite field of size $q$ and $\fq^*$ the set of non-zero elements of $\fq$. In this paper, we study a class of twisted generalized Reed-Solomon code $\mc_\ell(D, k, \eta, \vec{v})\subset \fq^n$ generated by the following matrix
\[
\left(\begin{array}{cccc}
v_{1} & v_{2} & \cdots & v_{n} \\
v_{1} \alpha_{1} & v_{2} \alpha_{2} & \cdots & v_{n} \alpha_{n} \\
\vdots & \vdots & \ddots & \vdots \\
v_{1} \alpha_{1}^{\ell-1} & v_{2} \alpha_{2}^{\ell-1} & \cdots & v_{n} \alpha_{n}^{\ell-1} \\
v_{1} \alpha_{1}^{\ell+1} & v_{2} \alpha_{2}^{\ell+1} & \cdots & v_{n} \alpha_{n}^{\ell+1} \\
\vdots & \vdots & \ddots & \vdots \\
v_{1} \alpha_{1}^{k-1} & v_{2} \alpha_{2}^{k-1} & \cdots & v_{n} \alpha_{n}^{k-1} \\
v_{1}\left(\alpha_{1}^{\ell}+\eta\alpha_{1}^{q-{2}}\right) & v_{2}\left(\alpha_{2}^{\ell}+ \eta
\alpha_{2}^{q-2}\right) &\cdots & v_{n}\left(\alpha_{n}^{\ell}+\eta\alpha_{n}^{q-2}\right)
\end{array}\right)
\]
where $0\leq \ell\leq k-1,$ the evaluation set $D=\{\alpha_{1},\alpha_{2},\cdots, \alpha_{n}\}\subseteq \fq^*$, scaling  vector $\vec{v}=(v_1,v_2,\cdots,v_n)\in (\fq^*)^n$ and $\eta\in\fq^*$. The minimum distance and dual code of $\mc_\ell(D, k, \eta, \vec{v})$ will be determined.  
For the special case $\ell=k-1,$ a sufficient and necessary condition for $\mc_{k-1}(D, k, \eta, \vec{v})$ to be self-dual will be given. We will also show that the code is MDS or near-MDS. Moreover, a complete classification when the code is near-MDS or MDS will be presented. 

\end{abstract}

\section{Introduction}

Let $\fq$ be a finite field of size $q$. It is well-known that parameters $[n,k,d]$ of any linear code over the finite filed $\fq$ obey the Singleton bound $d\leq n-k+1.$ A linear code with parameters $[n,k,d]$ is called maximum distance separable (MDS) if the parameters satisfy $d=n-k+1.$ For the near-optimal case $d=n-k$, the linear code is called almost-MDS. Moreover, if a linear code and its dual code are almost-MDS at the same time, then the linear code is called near-MDS. Since MDS codes and near-MDS codes play important roles in coding theory and have many applications, the study of MDS codes and near-MDS codes, including classification problem, constructions and those with self-dual property, has attracted a lot of attention~\cite{BGGHK,DL94,DL00,FLLL,GK,GG,GKL08,HK06,JK19,KL04,KKS09,kk15,RL89,SQS18,SSS18,WHL21}. 
Generalized Reed-Solomon (GRS) codes form a very important class of MDS codes. Due to the easy encoding and fast decoding for few burst errors, they are used in many communication system. The decoding performance of GRS codes is always a very important issue in the theoretical computer science. In recent years, constructions of self-dual MDS codes via GRS codes become a hot topic~\cite{FF19,FZX21,JX17,LLL19,Yan19,ZF20}. After the twisted GRS (TGRS) codes were introduced in~\cite{BPN}, the properties of TGRS codes and constructions of self-dual TGRS codes are studied extensively~\cite{BGGHK,HYN20,HYN21,LR20,LL21,ZL21}.

 Let $\fq^*$ be the set of non-zero elements of $\fq$. 
For any subset $D\subseteq \fq$ of size $n$ and any vector $\vec{v}=(v_1,v_2,\cdots,v_n)\in (\fq^*)^n$, the GRS code $GRS(D, k, \vec{v})$ is generated by 
\[
G_k(D,\vec{v})=\left(\begin{array}{cccc}
v_{1} & v_{2} & \cdots & v_{n} \\
v_{1} \alpha_{1} & v_{2} \alpha_{2} & \cdots & v_{n} \alpha_{n} \\
\vdots & \vdots & \ddots & \vdots \\
v_{1} \alpha_{1}^{k-1} & v_{2} \alpha_{2}^{k-1} & \cdots & v_{n} \alpha_{n}^{k-1} 
\end{array}\right),
\]
over the finite field $\fq$. It is not hard to prove that GRS codes are MDS codes. Let $G(x)=\prod_{\alpha\in D}(x-\alpha)\in \fq[x]$. Denote by $G'(x)$ the formal derivative of $G(x)$, i.e., $G'(x)=\sum_{i=1}^{n}\prod_{j\neq i}(x-\alpha_{j})$. Let $\vec{u}=(u_1,u_2,\cdots, u_n)\in(\fq^*)^n$ be defined by $u_i=\frac{1}{G'(\alpha_i)}$. Then the dual of the GRS code $GRS(D, k, \vec{v})$ is another GRS code~\cite{Huf03,Mac77} with generator matrix $G_{n-k}(D, \vec{u}\odot \vec{v}^{-1})$ where 
\[
\vec{u}\odot \vec{v}^{-1}=\left(\frac{u_1}{v_1}, \frac{u_2}{v_2}, \cdots, \frac{u_n}{v_n}\right).
\]
Equivalently, we have equalities $\sum_{i=1}^nu_i\alpha_i^j=0$ for all $j=0,1,\cdots, n-2$.

In~\cite{BPN}, the authors generalized the definition of GRS codes to TGRS codes. 
In this paper, we consider the following TGRS codes. For any vector $\vec{v}=(v_1,v_2,\cdots,v_n)\in (\fq^*)^n$, any subset  $D\subseteq \fq^*$ of size $n$, any integer $2\leq k\leq n-1$, any integer $0\leq \ell\leq k-1$ and any $\eta\in \fq^*$, the TGRS code $\mc_\ell(D, k, \eta, \vec{v})$ is generated by the following matrix
\[
G_{k,\ell}(D,\eta,\vec{v})=\left(\begin{array}{cccc}
v_{1} & v_{2} & \cdots & v_{n} \\
v_{1} \alpha_{1} & v_{2} \alpha_{2} & \cdots & v_{n} \alpha_{n} \\
\vdots & \vdots & \ddots & \vdots \\
v_{1} \alpha_{1}^{\ell-1} & v_{2} \alpha_{2}^{\ell-1} & \cdots & v_{n} \alpha_{n}^{\ell-1} \\
v_{1} \alpha_{1}^{\ell+1} & v_{2} \alpha_{2}^{\ell+1} & \cdots & v_{n} \alpha_{n}^{\ell+1} \\
\vdots & \vdots & \ddots & \vdots \\
v_{1} \alpha_{1}^{k-1} & v_{2} \alpha_{2}^{k-1} & \cdots & v_{n} \alpha_{n}^{k-1} \\
v_{1}\left(\alpha_{1}^{\ell}+\eta\alpha_{1}^{q-{2}}\right) & v_{2}\left(\alpha_{2}^{\ell}+ \eta
\alpha_{2}^{q-2}\right) &\cdots & v_{n}\left(\alpha_{n}^{\ell}+\eta\alpha_{n}^{q-2}\right)
\end{array}\right).
\]
Note that if $\eta=0$ in the above generator matrix, then the corresponding linear code is the GRS code $GRS(D, k, \vec{v})$. So we choose non-zero $\eta$ in the definition of the TGRS code. 
The main task in this paper is to study the properties of the TGRS code $\mc_\ell(D, k, \eta, \vec{v})$ such as the minimum distance, the dual code, conditions to be self-dual and near-MDS or MDS. 

We make a convention that all the notations above apply to the whole paper.

The rest of this paper is organized as follows. In Section~\ref{Sec:Dual}, the dual code of the TGRS code $\mc_\ell(D, k, \eta, \vec{v})$ is given explicitly and the self-dual property for the case $\ell=k-1$ is completely characterized. In Section~\ref{Sec:mindist}, the minimum distance of the TGRS code $\mc_\ell(D, k, \eta, \vec{v})$ is computed. Based on the computation, the TGRS code $\mc_{k-1}(D, k, \eta, \vec{v})$ is near-MDS or MDS if and only if certain subset product problem on the finite field is solvable or not. In Section~\ref{Sec:conc}, we conclude this paper.

\section{ The Dual Code of the TGRS code $\mc_\ell(D, k, \eta, \vec{v})$}\label{Sec:Dual}
In this section, we determine the dual code of the TGRS code $\mc_\ell(D, k, \eta, \vec{v})$. Generally, the TGRS code $\mc_\ell(D, k, \eta, \vec{v})$ is not self-dual. A sufficient and necessary condition is given for the TGRS code $\mc_{k-1}(D, k, \eta, \vec{v})$ being self-dual.

\begin{theorem}\label{dualcode}
	Let $b_0=1$ and $b_1,b_2,\cdots,b_{k-\ell-1}\in\fq$ be defined by the following recursion
	\begin{align}\label{recursion}
     b_j=-\frac{\sum_{r=0}^{j-1}b_r\sum_{i=1}^nu_i \alpha_i^{n+j-1-r}}{\sum_{i=1}^nu_i \alpha_i^{n-1}},\quad j=1,2,\cdots,k-\ell-1.
	\end{align}
	The TGRS code $\mc_\ell(D, k, \eta, \vec{v})$ has a parity-check matrix
	\[
	\left(\begin{array}{ccc}
	\frac{u_{1}}{v_{1}} f\left(\alpha_{1}\right) & \frac{u_{2}}{v_{2}} f\left(\alpha_{2}\right) \,\,\,\qquad \cdots & \frac{u_{n}}{v_{n}} f\left(\alpha_{n}\right) \\
	\hline\\
	 &G_{n-k-1}(D,\vec{u}\odot\vec{v}^{-1}\odot\vec{\alpha})
	\end{array}\right)
	\]
	 where $f(x)=x^{n-\ell-1}+b_1x^{n-\ell-2}+\cdots+b_{k-\ell-1}x^{n-k}-\frac{\sum_{i=1}^nu_i \alpha_i^{n-1}}{\eta\sum_{i=1}^nu_i \alpha_i^{-1}}.$
\end{theorem}

\begin{remark}
	The denominators  $\sum_{i=1}^nu_i \alpha_i^{-1}, \sum_{i=1}^nu_i \alpha_i^{n-1}$ in the above theorem are non-zero, so the fractions above do make sense. 
	
 If $\sum_{i=1}^nu_i \alpha_i^{-1}=0$, then $(u_1, u_2,\cdots, u_n)$ is a solution of the system of linear equations
\[
\left(\begin{array}{cccc}
\alpha_{1}^{-1} & \alpha_{2}^{-1} & \cdots & \alpha_{n}^{-1} \\
1 & 1 & \cdots & 1\\
\alpha_{1} & \alpha_{2} & \cdots & \alpha_{n} \\
\vdots & \vdots & \ddots & \vdots \\
\alpha_{1}^{n-2} & \alpha_{2}^{n-2} & \cdots & \alpha_{n}^{n-2} 
\end{array}\right)X^T=0.
\]
The coefficient matrix is non-degenerated by Vandermonde determinant formula. So the system has only zero solution, which contradicts to that $(u_1, u_2,\cdots, u_n)$ is non-zero. So $\sum_{i=1}^nu_i \alpha_i^{-1}\neq 0$. By using the same argument, one can show that $\sum_{i=1}^nu_i \alpha_i^{n-1}\neq 0$.

\end{remark}

\begin{proof}
Note that 
\[
GRS(D,\ell,\vec{v})\subsetneq\mc_\ell(D, k, \eta, \vec{v}) \subsetneq GRS(D,k+1,\vec{v}\odot \vec{\alpha}^{-1}) 
\]
where $\vec{\alpha}=(\alpha_{1}, \alpha_{2},\cdots, \alpha_n)$. So we have 
\[
GRS(D,k+1,\vec{v}\odot \vec{\alpha}^{-1}) ^\perp\subsetneq \mc_\ell(D, k, \eta, \vec{v}) ^\perp\subsetneq GRS(D,\ell,\vec{v})^{\perp}.
\]
The code $GRS(D,k+1,\vec{v}\odot \vec{\alpha}^{-1}) ^\perp$ has a generator matrix 
\[
G_{n-k-1}(D,\vec{u}\odot\vec{v}^{-1}\odot\vec{\alpha})=\left(\begin{array}{cccc}
\frac{u_1}{v_1}\alpha_1 & \frac{u_2}{v_2}\alpha_2 & \cdots & \frac{u_n}{v_n}\alpha_{n} \\
\frac{u_1}{v_1} \alpha_{1}^2 & \frac{u_2}{v_2}\alpha_{2}^2 & \cdots & \frac{u_n}{v_n}\alpha_{n}^2 \\
\vdots & \vdots & \ddots & \vdots \\
\frac{u_1}{v_1}\alpha_{1}^{n-k-1} & \frac{u_2}{v_2} \alpha_{2}^{n-k-1} & \cdots & \frac{u_n}{v_n} \alpha_{n}^{n-k-1} 
\end{array}\right).
\]

It is easy to see that $GRS(D,k+1,\vec{v}\odot \vec{\alpha}^{-1}) ^\perp$ has codimension $1$ in $ \mc_\ell(D, k, \eta, \vec{v}) ^\perp$. 
Since $\mc_\ell(D, k, \eta, \vec{v}) ^\perp\subsetneq GRS(D,\ell,\vec{v})^{\perp}=GRS(D,n-\ell,\vec{u}\odot \vec{v}^{-1})$, we may consider non-zero polynomials of the form $f(x)=b_0x^{n-1-\ell}+b_1x^{n-\ell-2}+\cdots+b_{k-\ell-1}x^{n-k}+b\in\fq[x]$ (here, the terms $x^{n-k-1}, x^{n-k-2},\cdots, x$ are absorbed in $GRS(D,k+1,\vec{v}\odot \vec{\alpha}^{-1}) ^\perp$) where $b_0,b_1,\cdots, b_{k-\ell-1}$ and $b$ are to be determined. 

On one hand, the vector $\left(\frac{u_{1}}{v_{1}} f\left(\alpha_{1}\right) ,\frac{u_{2}}{v_{2}} f\left(\alpha_{2}\right),\cdots ,\frac{u_{n}}{v_{n}} f\left(\alpha_{n}\right)\right) $ does not belong to $GRS(D,k+1,\vec{v}\odot \vec{\alpha}^{-1}) ^\perp$. If not, there is a polynomial $A(x)=a_1x+a_2x^2+\cdots+a_{n-k-1}x^{n-k-1}\in \fq[x]$ such that
$\frac{u_i}{v_i}f(\alpha_i)=\frac{u_i}{v_i}A(\alpha_i)$ for all $i=1,2,\cdots,n$ which implies that the polynomial $f(x)-A(x)$ has at least $n$ different roots. But the degree of $f(x)-A(x)$ is at most $n-1-\ell\leq n-1.$ So as polynomials, $f(x)=A(x)$ which is impossible! 

On the other hand, the vector $\left(\frac{u_{1}}{v_{1}} f\left(\alpha_{1}\right) ,\frac{u_{2}}{v_{2}} f\left(\alpha_{2}\right),\cdots ,\frac{u_{n}}{v_{n}} f\left(\alpha_{n}\right)\right) $ belongs to $\mc_\ell(D, k, \eta, \vec{v}) ^\perp$ if and only if the following system of equalities holds
\[
\begin{cases}
\sum_{i=1}^n\frac{u_{i}}{v_{i}} f\left(\alpha_{i}\right)v_i\alpha_{i}^{\ell+1}=0\\
\sum_{i=1}^n\frac{u_{i}}{v_{i}} f\left(\alpha_{i}\right)v_i\alpha_{i}^{\ell+2}=0\\
\cdots\\
\sum_{i=1}^n\frac{u_{i}}{v_{i}} f\left(\alpha_{i}\right)v_i\alpha_{i}^{k-1}=0\\
\sum_{i=1}^n\frac{u_{i}}{v_{i}} f\left(\alpha_{i}\right)v_i \left(\alpha_{i}^{\ell}+\eta\alpha_{i}^{q-2}\right)=0.
\end{cases}
\]
Since $\alpha_{i}\in \fq^*$, we have $\alpha_{i}^{q-2}=\alpha_{i}^{-1}$ for all $i=1,2,\cdots, n$. So it follows from the equalities $$\sum_{i=1}^nu_i\alpha_i^j=0,\,\forall j=0,1,\cdots, n-2$$ that 
\[
\begin{cases}
b_0\sum_{i=1}^nu_i \alpha_i^{n}+b_1\sum_{i=1}^nu_i \alpha_i^{n-1}=0\\
b_0\sum_{i=1}^nu_i \alpha_i^{n+1}+b_1\sum_{i=1}^nu_i \alpha_i^{n}+b_2\sum_{i=1}^nu_i \alpha_i^{n-1}=0\\
\cdots\\
b_0\sum_{i=1}^nu_i \alpha_i^{n+k-\ell-2}+b_1\sum_{i=1}^nu_i \alpha_i^{n+k-\ell-3}+\cdots+b_{k-\ell-1}\sum_{i=1}^nu_i \alpha_i^{n-1}=0\\
b_0\sum_{i=1}^nu_i \alpha_i^{n-1}+b\eta\sum_{i=1}^nu_i \alpha_i^{-1}=0.
\end{cases}
\]

Note that $b_0\neq 0$, so we can assume $b_0=1$ by linearity. If $b_0=0$, then it follows from the first and last equalities that $b_1=b=0$ since $\sum_{i=1}^nu_i \alpha_i^{n-1}\neq 0$, $\sum_{i=1}^nu_i \alpha_i^{-1}\neq 0$ and $\eta\neq 0$. As a consequence of $b_0=b_1=0$, we have $b_2=0$ from the second equality. Similarly, we can get $b_3=\cdots=b_{k-\ell-1}=0$ and hence $f(x)=0$ which contradicts to the assumption that $f(x)$ is non-zero.

So by solving the above system of equations and by assumption $b_0=1$, we can obtain that elements $b_1, b_2,\cdots, b_{k-\ell-1}$ indeed satisfy the recursive condition~(\ref{recursion}) and 
\[
b=-\frac{\sum_{i=1}^nu_i \alpha_i^{n-1}}{\eta\sum_{i=1}^nu_i \alpha_i^{-1}}.
\]

\end{proof}

\begin{corollary}\label{dual:k-1}
	The dual of the TGRS code $\mc_{k-1}(D, k, \eta, \vec{v})$ is another TGRS code $\mc_{n-k-1}(D, n-k, \eta', \vec{u}\odot\vec{v}^{-1}\odot\vec{\alpha})$
	where $\eta'=-\frac{\sum_{i=1}^nu_i \alpha_i^{n-1}}{\eta\sum_{i=1}^nu_i \alpha_i^{-1}}.$
\end{corollary}
\begin{proof}
By Theorem~\ref{dualcode}, the dual of the TGRS code $\mc_{k-1}(D, k, \eta, \vec{v})$ has a generator matrix of the following form
\[
\left(\begin{array}{cccc}
\frac{u_1}{v_1}(\alpha_1^{n-k}+\eta') & \frac{u_2}{v_2}(\alpha_2^{n-k}+\eta') & \cdots & \frac{u_n}{v_n}(\alpha_n^{n-k}+\eta') \\
\frac{u_1}{v_1}\alpha_1 & \frac{u_2}{v_2}\alpha_2 & \cdots & \frac{u_n}{v_n}\alpha_{n} \\
\frac{u_1}{v_1} \alpha_{1}^2 & \frac{u_2}{v_2}\alpha_{2}^2 & \cdots & \frac{u_n}{v_n}\alpha_{n}^2 \\
\vdots & \vdots & \ddots & \vdots \\
\frac{u_1}{v_1}\alpha_{1}^{n-k-1} & \frac{u_2}{v_2} \alpha_{2}^{n-k-1} & \cdots & \frac{u_n}{v_n} \alpha_{n}^{n-k-1} 
\end{array}\right),
\]
where $\eta'=-\frac{\sum_{i=1}^nu_i \alpha_i^{n-1}}{\eta\sum_{i=1}^nu_i \alpha_i^{-1}}.$ We can rewrite the above generator matrix as following
\[
\left(\begin{array}{cccc}
\frac{u_1\alpha_1}{v_1}1 & \frac{u_2\alpha_2}{v_2} 1 & \cdots & \frac{u_n\alpha_n}{v_n}1 \\
\frac{u_1\alpha_1}{v_1}\alpha_{1} & \frac{u_2\alpha_2}{v_2} \alpha_{2} & \cdots & \frac{u_n\alpha_n}{v_n}\alpha_{n} \\
\vdots & \vdots & \ddots & \vdots \\
\frac{u_1\alpha_1}{v_1}\alpha_{1}^{n-k-2} & \frac{u_2\alpha_2}{v_2} \alpha_{2}^{n-k-2} & \cdots & \frac{u_n\alpha_n}{v_n} \alpha_{n}^{n-k-2} \\
\frac{u_1\alpha_1}{v_1}(\alpha_1^{n-k-1}+\eta'\alpha_1^{-1}) & \frac{u_2\alpha_2}{v_2}(\alpha_2^{n-k-1}+\eta'\alpha_2^{-1}) & \cdots & \frac{u_n\alpha_n}{v_n}(\alpha_n^{n-k-1}+\eta'\alpha_n^{-1}) \\
\end{array}\right).
\]
Since $\alpha_i\in \fq^*$, we have $\alpha_i^{-1}=\alpha_i^{q-2}$ for all $i=1,2,\cdots, n$. So the dual of the TGRS code $\mc_{k-1}(D, k, \eta, \vec{v})$ is the TGRS code $\mc_{n-k-1}(D, n-k, \eta', \vec{u}\odot\vec{v}^{-1}\odot\vec{\alpha})$.
\end{proof}

Now, we can determine when the TGRS code $\mc_{k-1}(D, k, \eta, \vec{v})$  is self-dual. 
\begin{theorem}\label{selfdual}
	For $n=2k,$ the TGRS code $\mc_{k-1}(D, k, \eta, \vec{v})$  is self-dual if and only if the following two conditions hold:
	\begin{enumerate}
		\item there exists some $\lambda\in\fq^*$ such that ${v_i^2}=\lambda{u_i\alpha_{i}},\,\forall i=1,2,\cdots,n$;
		\item $\eta^2=-\frac{\sum_{i=1}^nu_i \alpha_i^{n-1}}{\sum_{i=1}^nu_i \alpha_i^{-1}}.$
	\end{enumerate}
\end{theorem}
\begin{proof}
(The necessary direction ``$\implies$'')
Recall that the TGRS code $\mc_{k-1}(D, k, \eta, \vec{v})$ and its dual have generator matrices
\[
\left(\begin{array}{cccc}
v_{1} & v_{2} & \cdots & v_{n} \\
v_{1} \alpha_{1} & v_{2} \alpha_{2} & \cdots & v_{n} \alpha_{n} \\
\vdots & \vdots & \ddots & \vdots \\
v_{1} \alpha_{1}^{k-2} & v_{2} \alpha_{2}^{k-2} & \cdots & v_{n} \alpha_{n}^{k-2} \\
v_{1}\left(\alpha_{1}^{k-1}+\eta\alpha_{1}^{q-{2}}\right) & v_{2}\left(\alpha_{2}^{k-1}+ \eta
\alpha_{2}^{q-2}\right) &\cdots & v_{n}\left(\alpha_{n}^{k-1}+\eta\alpha_{n}^{q-2}\right)
\end{array}\right)
\]
and 
\[
\left(\begin{array}{cccc}
\frac{u_1}{v_1}(\alpha_1^k+\eta')& \frac{u_2}{v_2}(\alpha_2^k+\eta') & \cdots & \frac{u_n}{v_n}(\alpha_n^k+\eta')  \\
\frac{u_1}{v_1}\alpha_1 & \frac{u_2}{v_2}\alpha_2 & \cdots & \frac{u_n}{v_n}\alpha_{n} \\
\frac{u_1}{v_1} \alpha_{1}^2 & \frac{u_2}{v_2}\alpha_{2}^2 & \cdots & \frac{u_n}{v_n}\alpha_{n}^2 \\
\vdots & \vdots & \ddots & \vdots \\
\frac{u_1}{v_1}\alpha_{1}^{k-1} & \frac{u_2}{v_2} \alpha_{2}^{k-1} & \cdots & \frac{u_n}{v_n} \alpha_{n}^{k-1} 
\end{array}\right),
\]
respectively, where $\eta'=-\frac{\sum_{i=1}^nu_i \alpha_i^{n-1}}{\eta\sum_{i=1}^nu_i \alpha_i^{-1}}$. So if the TGRS code $\mc_{k-1}(D, k, \eta, \vec{v})$  is self-dual, then 
\begin{itemize}
	\item there exists $\left(a_{0}, a_{1}, \cdots, a_{k-1}\right) \in \mathbb{F}_{q}^{k} $  such that 
	\[
	\frac{v_{i}^{2}}{u_{i}}=a_{0}\left(\alpha_{i}^{k}+\eta^{\prime}\right)+a_{1} \alpha_{i}+\cdots+a_{k-1} \alpha_{i}^{k-1}, \quad i=1,2, \cdots, n;
	\]
    \item there exists $\left(b_{0}, b_{1}, \cdots, b_{k-1}\right) \in \mathbb{F}_{q}^{k} $  such that 
    \[
     \frac{v_{i}^{2}}{u_{i}}\alpha_{i}^{k-2}=b_{0}\left(\alpha_{i}^{k}+\eta^{\prime}\right)+b_{1} \alpha_{i}+\cdots+b_{k-1} \alpha_{i}^{k-1}, \quad i=1,2, \cdots, n.
    \]
\end{itemize}
For any vector $\vec{a}=\left(a_{0}, a_{1}, \cdots, a_{k-1}\right) \in \mathbb{F}_{q}^{k} $, let
\[
h_{\vec{a}}(x)=a_{0}\left(x^{k}+\eta^{\prime}\right)+a_{1} x+\cdots+a_{k-1} x^{k-1}\in\fq[x].
\]
Then the polynomial $h_{\vec{a}}(x)x^{k-2}-h_{\vec{b}}(x)$ has $n$ roots $\alpha_{1}, \alpha_{2}, \cdots, \alpha_{n}$ and degree at most $2k-2=n-2$. So  $h_{\vec{a}}(x)x^{k-2}-h_{\vec{b}}(x)$ is the zero polynomial. By looking at the coefficients, we get that $a_0=a_2=a_3=\cdots=a_{k-1}=0$, $b_0=b_1=\cdots=b_{k-2}=0$, and $a_1=b_{k-1}\in \fq^*.$ Suppose $a_1=b_{k-1}=\lambda\in \fq^*$. Then
\[
\frac{v_i^2}{u_i}=\lambda\alpha_{i}\,\,\mbox{or}\,\,{v_i^2}=\lambda{u_i\alpha_{i}}\,\,\forall i=1,2,\cdots,n.
\]
So the dual code  $\mc_{k-1}(D, k, \eta, \vec{v})^\perp$ has a generator matrix 
\[
\left(\begin{array}{cccc}
v_{1}\left(\alpha_{1}^{k-1}+\eta'\alpha_{1}^{-1}\right) & v_{2}\left(\alpha_{2}^{k-1}+ \eta'
\alpha_{2}^{-1}\right) &\cdots & v_{n}\left(\alpha_{n}^{k-1}+\eta'\alpha_{n}^{-1}\right)\\
v_{1} & v_{2} & \cdots & v_{n} \\
v_{1} \alpha_{1} & v_{2} \alpha_{2} & \cdots & v_{n} \alpha_{n} \\
\vdots & \vdots & \ddots & \vdots \\
v_{1} \alpha_{1}^{k-2} & v_{2} \alpha_{2}^{k-2} & \cdots & v_{n} \alpha_{n}^{k-2} \\
\end{array}\right).
\]
Since the TGRS code $\mc_{k-1}(D, k, \eta, \vec{v})$  is self-dual, we have $\eta'=\eta$. That is 
\[
\eta^2=-\frac{\sum_{i=1}^nu_i \alpha_i^{n-1}}{\sum_{i=1}^nu_i \alpha_i^{-1}}.
\]

\begin{flushleft}
(The sufficient part ``$\Longleftarrow$'') It is obvious from the proof above.
\end{flushleft}

\end{proof}

Note that the Frobenius map is a permutation on finite fields of even characteristic. So we have the following corollary.
\begin{corollary}
	Let $\fq$  be a finite field of characteristic $2$. Let $D=\{\alpha_{1},\alpha_{2},\cdots, \alpha_{n}\}\subset \fq^*$ be any subset of size $n=2k$ and denote by $u_i=\frac{1}{\prod_{j\neq i}(\alpha_{i}-\alpha_{j})}$ and by $v_i=\sqrt{u_i\alpha_{i}}$ for $i=1,2,\cdots,n$. Denote by $\eta=\sqrt{\frac{\sum_{i=1}^nu_i \alpha_i^{n-1}}{\sum_{i=1}^nu_i \alpha_i^{-1}}}\in\fq^*.$ Then the TGRS code $\mc_{k-1}(D, k, \eta, \vec{v})$ over $\fq$ is self-dual.
	
\end{corollary}

\section{Minimum Distances of $\mc_\ell(D, k, \eta, \vec{v})$ and Its Dual}\label{Sec:mindist}
In this section, we investigate the minimum distances of the TGRS code $\mc_\ell(D, k, \eta, \vec{v})$ and its dual code $\mc_\ell(D, k, \eta, \vec{v})^{\perp}$. 

For any subset $S\subset \fq$ of size $k$ and for any integer $1\leq l\leq k$, we denote by $\sigma_{l}(S)=\sum_{T\subset S,\, \#T=l}\prod_{\beta\in T}\beta$ the $l$-th elementary symmetric polynomial on $S$. 

\begin{theorem}\label{mindist}
    Notations as above. We have the following:
	\begin{enumerate}
		\item The TGRS code $\mc_{\ell}(D, k, \eta, \vec{v})$ is MDS or almost-MDS.
		\item The TGRS code $\mc_{\ell}(D, k, \eta, \vec{v})$ is MDS if and only if there does not exist any subset $S\subset D$ of size $k$ such that $\eta=(-1)^{\ell+1}\frac{\sigma_k(S)}{\sigma_{k-1-\ell}(S)}.$
	\end{enumerate}
\end{theorem}
\begin{proof}
	For the first statement, it is easy to see that any $k$-dimensional subspace of an $[n,k+1]$-MDS code is MDS or almost-MDS. And the TGRS code $\mc_{\ell}(D, k, \eta, \vec{v})$ is a $k$-dimensional subspace of GRS code $GRS(D,k+1,\vec{v}\odot \vec{\alpha}^{-1})$, so $\mc_{\ell}(D, k, \eta, \vec{v})$ is MDS or almost-MDS. In order to prove the second statement, we give a new proof of the first statement.
	
	Let $d$ be the minimum distance of the TGRS code $\mc_\ell(D, k, \eta, \vec{v})$.
	Since the action $\vec{r}\mapsto \vec{r} \odot \vec{v}$ is Hamming-distance-preserving, we have 
	\begin{align*}
	d=&\min_{(a_0,a_1,\cdots, a_{k-1})\in \fq^k\setminus\{\vec{0}\}} \#\{\alpha\in D\mid a_0+a_1\alpha+\cdots+a_{k-1}\alpha^{k-1}+\eta a_l\alpha^{-1}\neq 0\}\\
	=&n-\max_{(a_0,a_1,\cdots, a_{k-1})\in \fq^k\setminus\{\vec{0}\}}\#\{\mbox{zeros of $a_0+a_1x+\cdots+a_{k-1}x^{k-1}+\eta a_lx^{-1}$ in D}\}\\
	=&n-\max_{(a_0,a_1,\cdots, a_{k-1})\in \fq^k\setminus\{\vec{0}\}}\#\{\mbox{zeros of $\eta a_l+a_0x+a_1x^2+\cdots+a_{k-1}x^{k} $ in D}\}\\
	\geq & n-k \label{d=n-k}\tag{2}
	\end{align*}
	where the last inequality follows from $\deg(\eta a_l+a_0x+a_1x^2+\cdots+a_{k-1}x^{k}) \leq k.$
	
	On the other hand, by the Singleton bound, we have $d\leq n-k+1.$ So 
	\[
	d\in \{n-k, n-k+1\}.
	\]
	In other words, the TGRS code $\mc_{\ell}(D, k, \eta, \vec{v})$ is almost-MDS or MDS. 

	Note that the equality holds in the inequality~(\ref{d=n-k}) if and only if 
	\[
	\max_{(a_0,a_1,\cdots, a_{k-1})\in \fq^k\setminus\{\vec{0}\}}\#\{\mbox{zeros of $\eta a_l+a_0x+a_1x^2+\cdots+a_{k-1}x^{k} $ in D}\}=k,
	\]
	which is equivalent to that there exists a subset $S\subset D$ of size $k$ such that
	\[
	\eta a_l+a_0x+a_1x^2+\cdots+a_{k-1}x^{k}=a_{k-1}\prod_{\alpha\in S}(x-\alpha).
	\]
	The last condition is equivalent to that there exists a subset $S\subset D$ of size $k$ such that
	 $$\eta=(-1)^{\ell+1}\frac{\prod_{\alpha\in S}\alpha}{\sum_{T\subset S,\, \#T=k-\ell-1}\prod_{\beta\in T}\beta}.$$
	
\end{proof}

In general, one can not replace almost-MDS by near-MDS. For the special case $\ell=k-1$, we have the following corollary.

\begin{corollary} \label{near-MDS}
	\begin{enumerate}
		\item The TGRS code $\mc_{k-1}(D, k, \eta, \vec{v})$ is MDS or near-MDS.
		\item The TGRS code $\mc_{k-1}(D, k, \eta, \vec{v})$ is near-MDS if and only if there exists a subset $S\subset D$ of size $k$ such that $\eta=(-1)^k\prod_{\alpha\in S}\alpha.$
	\end{enumerate}
\end{corollary}
\begin{proof}

By Theorem~\ref{mindist}, the TGRS code $\mc_{k-1}(D, k, \eta, \vec{v})$ is almost-MDS or MDS. Note that the dual code $\mc_{k-1}(D, k, \eta, \vec{v})^{\perp}$ which is also a TGRS code by Corollary~\ref{dual:k-1}. So by the same argument above,  the dual code $\mc_{k-1}(D, k, \eta, \vec{v})^{\perp}$ is almost-MDS or MDS. It is well-known that the dual code of any MDS code is still MDS. So  the TGRS code $\mc_{k-1}(D, k, \eta, \vec{v})$ and its dual code are almost-MDS or MDS simultaneously. That is, $\mc_{k-1}(D, k, \eta, \vec{v})$ is near-MDS or MDS. 

The second statement follows from Theorem~\ref{mindist} and the first statement of this corollary.

\end{proof}
	
\begin{remark}
	In~\cite{HYN21}, the authors showed that the code constructed there is near-MDS if and only if certain subset sum problem has a solution. Here, in our construction, the code is near-MDS if and only if the following subset product problem 
	\[
	 \mbox{find a subset $S\subset D$ of size $k$ such that $\eta=(-1)^k\prod_{\alpha\in S}\alpha$}
	\]
	has a solution.
\end{remark}	

\begin{corollary}
	The following statements are equivalent.
	\begin{enumerate}
		\item The TGRS code $\mc_{k-1}(D, k, \eta, \vec{v})$ is near-MDS.
		\item There exists a subset $S\subset D$ of size $k$ such that $\eta=(-1)^k\prod_{\alpha\in S}\alpha.$
		\item  There exists a subset $T\subset D$ of size $n-k$ such that $-\frac{\sum_{i=1}^nu_i \alpha_i^{n-1}}{\eta\sum_{i=1}^nu_i \alpha_i^{-1}}=(-1)^{n-k}\prod_{\alpha\in T}\alpha.$
	\end{enumerate}
\end{corollary}

%\section{Examples}
Next, we give an example to illustrate the above theorems.
\begin{example}
	Suppose the finite field $\fq$ has odd characteristic. Take the evaluation set $D=\fq^*\subset \mathbb{F}_{q^2}.$ In this case,  
	\[
	\sum_{i=1}^nu_i \alpha_i^{n-1}=\sum_{i=1}^nu_i \alpha_i^{q-2}=\sum_{i=1}^nu_i \alpha_i^{-1}.
	\]
    Let $\eta \in \mathbb{F}_{q^2}$ be such that $\eta^2=-1$. Moreover, elements $u_i\alpha_{i}\,(\forall \alpha_i\in D)$ are squares in $\mathbb{F}_{q^2}$, i.e., there exist $v_i\in \mathbb{F}_{q^2}$ such that $v_i^2=u_i\alpha_{i},\,i=1,2,\cdots,q-1$. By Theorem~\ref{selfdual}, the TGRS code $\mc_{\frac{q-3}{2}}(\fq^*, \frac{q-1}{2}, \eta, \vec{v})$ is a self-dual code over $\mathbb{F}_{q^2}$.
    
     If $q\equiv 3\pmod 4$, then the subset product problem 
    	\[
    \mbox{find a subset $S\subset \fq^*$ of size $k=\frac{q-1}{2}$ such that $\eta=(-1)^k\prod_{\alpha\in S}\alpha$}
    \]
    has no solution as $\eta \in \mathbb{F}_{q^2}\setminus\fq$. By Corollary~\ref{near-MDS}, $\mc_{\frac{q-3}{2}}(\fq^*, \frac{q-1}{2}, \eta, \vec{v})$ over $\mathbb{F}_{q^2}$ is MDS. %which is not equivalent to any GRS code.
    
     If $q\equiv 1\pmod 4$, then the subset product problem 
    \[
    \mbox{find a subset $S\subset \fq^*$ of size $k=\frac{q-1}{2}$ such that $\eta=(-1)^k\prod_{\alpha\in S}\alpha$}
    \]
    has a solution. Indeed, let $\zeta\in \fq^*$ be a primitive element of $\fq$ and $\eta=\zeta^{\frac{q-1}{4}}$ a square root of $-1$. Then the above subset product problem is equivalent to finding a subset $S\subset \{0,1,\cdots,q-2\}$  of size $\frac{q-1}{2}$ satisfying $\sum_{s\in S}s\equiv \frac{3(q-1)}{4} \pmod {q-1}.$ It is sufficient by taking 
    \[
    S=\begin{cases}
    \{1,2,\cdots, \frac{q-1}{2}\},& \mbox{if $q\equiv 5\pmod 8$;}\\
    \{0,1,\cdots, \frac{q-3}{2}\},& \mbox{if $q\equiv 1\pmod 8$.}
    \end{cases}
    \]
     By Corollary~\ref{near-MDS}, $\mc_{\frac{q-3}{2}}(\fq^*, \frac{q-1}{2}, \eta, \vec{v})$ over $\mathbb{F}_{q^2}$ is near-MDS.
\end{example}

\section{Conclusion}\label{Sec:conc}
In this paper, a class of TGRS codes were studied. The properties such as self-dual, near-MDS, MDS of the codes were considered. It is interesting that some special TGRS codes are near-MDS or MDS if and only if certain subset product problem is solvable or not. The subset sum problem on finite fields has attracted a lot of attention in the last decades, including theoretical aspect and applications in cryptography and coding theory, etc. Although the subset product problem on finite fields $\fq$ is equivalent to the subset sum problem on the residue class ring $\mathbb{Z}/(q-1)\mathbb{Z}$ by taking discrete logarithm, their computational hardness might be slightly different as discrete logarithm problem over finite fields is generally hard. It might be interesting to study the subset product problem on finite fields directly.

\end{document}